\pdfoutput=1
\documentclass[sigconf]{acmart}

\usepackage{algorithm}
\usepackage[noend]{algpseudocode}
\algrenewcomment[1]{\hfill \(\triangleright\) \emph{\small #1}} 
\algnewcommand{\CommentLine}[1]{\(\triangleright\) \emph{\small #1}}
\algnewcommand{\InlineIf}[2]{
  \algorithmicif\ #1\ \algorithmicthen\ #2}
\algnewcommand{\InlineFor}[2]{\algorithmicfor\ #1\ \algorithmicdo\ #2} 

\usepackage[capitalise]{cleveref}

\newcommand{\bigO}[1]{\mathchoice{O\left(#1\right)}{O(#1)}{O(#1)}{O(#1)}} 
\newcommand{\timepm}[1]{\mathchoice{\mathsf{M}\left(#1\right)}{\mathsf{M}(#1)}{\mathsf{M}(#1)}{\mathsf{M}(#1)}} 
\newcommand{\ZZ}{\mathbb{Z}} 
\newcommand{\NN}{\mathbb{N}} 
\newcommand{\field}{\mathbb{K}} 
\newcommand{\K}{\field} 
\newcommand{\LL}{\mathbb{L}} 
\newcommand{\pring}{\field[x]} 
\newcommand{\ffield}{\field(x)} 
\newcommand{\closure}[1]{\overline{#1}} 
\newcommand{\ps}[1]{[\hspace{-0.09cm}[#1]\hspace{-0.08cm}]}
\newcommand{\seq}[1]{(#1_n)_{n\geq0}}
\newcommand{\pbseqterm}{\textsc{SeqTerm}} 
\newcommand{\pbbivmodpow}{\textsc{BivModPow}} 
\newcommand{\pbpolmatpow}{\textsc{PolMatPow}} 
\newcommand{\myparagraph}[1]{\smallskip\emph{#1.}} 

\author{Alin Bostan}
\affiliation{%
  \institution{Inria}
  \city{Palaiseau}
  \country{France}%
}

\author{Vincent Neiger}
\affiliation{%
	\institution{Sorbonne Universit\'e, \textsc{CNRS}, \textsc{LIP6}}
	\city{F-75005 Paris}
	\postcode{75252}\country{France}}

\author{Sergey Yurkevich}
\affiliation{%
  \institution{University of Vienna and Inria Saclay}
  \city{Vienna}
  \country{Austria}%
}

\title{Beating binary powering for polynomial matrices}

\copyrightyear{2023}
\acmYear{2023}
\setcopyright{acmlicensed}\acmConference[ISSAC 2023]{International Symposium on Symbolic and Algebraic Computation 2023}{July 24--27, 2023}{Tromsø, Norway}
\acmBooktitle{International Symposium on Symbolic and Algebraic Computation 2023 (ISSAC 2023), July 24--27, 2023, Tromsø, Norway}
\acmPrice{15.00}
\acmDOI{10.1145/3597066.3597118}
\acmISBN{979-8-4007-0039-2/23/07}



\begin{document}

\newtheorem{remark}[theorem]{Remark}

\begin{abstract}
The $N$th power of a polynomial matrix of fixed size and degree
can be computed by binary powering as fast as
multiplying two polynomials of linear degree in~$N$.
When Fast Fourier Transform (FFT) is available, the resulting
complexity is \emph{softly linear} in~$N$,
i.e.~linear in~$N$ with extra logarithmic factors.
We show that it is possible to beat binary powering,
by an algorithm whose complexity is \emph{purely linear} in~$N$,
even in absence of FFT.
The key result making this improvement possible
is that the entries of the $N$th power
of a polynomial matrix
satisfy linear differential equations with polynomial coefficients
whose orders and degrees are independent of~$N$.
Similar algorithms are proposed for two related problems:
computing the $N$th term of a C-finite sequence of polynomials,
and modular exponentiation to the power $N$ for bivariate polynomials.
\end{abstract}

\thanks{%
  The authors thank Bruno Salvy for his remarks, and the anonymous
  referees for their reports. The authors are supported by the French project \textsc{De Rerum Natura} (ANR-19-CE40-0018)
 and by the joint French–Austrian project \textsc{EAGLES} 
 (ANR-22-CE91-0007 \& FWF I6130-N). 
 The third author was supported by the
\href{https://www.oeaw.ac.at/en/1/austrian-academy-of-sciences}{ÖAW} DOC
fellowship P-26101.}

\maketitle

\section{Introduction}
\label{sec_intro}

A sequence $(u_n)_{n\geq 0}$ is called \emph{C-finite} if it satisfies a
linear recurrence relation whose coefficients are constant with respect to~$n$.
The famous sequence $(f_n)_{n \geq 0}$ of Fibonacci numbers, defined by the
recurrence $f_{n+2} = f_{n+1} + f_n$ and the initial values $f_0 = 0, f_1=1$,
is perhaps the most basic example of a C-finite sequence after the geometric
ones $(q^n)_{n \geq 0}$. It is classical that the term $f_N$ can be computed in
$O(\log(N))$ arithmetic operations, thus as fast as $q^N$. This can be achieved
by \emph{binary powering} for \(q^N\), and in fact for $f_N$ as well, since
it is the top-right entry of $C^N$ where $C$
is the \(2 \times 2\) companion matrix
$(\begin{smallmatrix} 0 & 1 \\ 1 & 1
\end{smallmatrix})$.
This idea generalizes to any C-finite sequence $(u_n)_{n\geq 0}$: a
recurrence of order $r\geq 1$ for $(u_n)_{n\geq 0}$ can be encoded, via its
\textit{companion matrix}, into an $r\times r$ matrix recurrence of
order~\(1\).  Then the term $u_N$ of the sequence appears as the first entry of 
the product of
the vector
of initial values $(u_0 , \dots , u_{r-1})$ by the
\(N\)th power of this \(r \times r\) companion matrix~\cite{MiBr66,Fiduccia85}. 
Then $u_N$ can be computed in $O(\log(N))$ arithmetic operations, and in \(O(N
\log(N))\) bit operations if $(u_n)_{n\geq 0}$ is an integer sequence, using
fast integer multiplication \cite{HaHo21}. Here \(r\) is considered
constant, i.e., \(r \in \bigO{1}\).

Fibonacci polynomials $F_n(x)$ are a natural generalization of
Fibonacci numbers (see e.g. ~\cite{Byrd63}).
 They are defined by the recurrence
\begin{equation} \label{eq:def_Fib}
  F_{n+2}(x) = x  F_{n+1}(x) + F_n(x) \quad \text{for } n \geq 0
\end{equation}
and the initial values $F_0(x) = 0, F_1(x)=1$. The first few terms are
\(
  \seq{F} = (0,1,x, x^2+1, x^3+2x, x^4+3x^2+1, \dots ).
\)
Obviously, for all $n\geq 1$, the polynomial $F_n(x)$ is monic of degree $n-1$
and the sum of its coefficients is $F_n(1) = f_n$.

Given $N \in \NN$, the direct iterative algorithm for computing $F_N(x)$ has
complexity $O(N^2)$. It computes, for each \(n\le N\), all the \(n\)
coefficients of the intermediate polynomial $F_n(x)$; in total this amounts to
$\Theta(N^2)$ coefficients. Therefore, if one wants to compute all of
\((F_0,\ldots,F_N)\) then this direct method is optimal with respect to the
total arithmetic size of the output. However, it becomes quadratic if one is
only interested in determining $F_N(x)$ alone.

To compute the polynomial $F_N(x)$ faster, one can use,
as in the scalar case, the
reformulation of the second-order recurrence (\ref{eq:def_Fib}) as a
first-order (polynomial) matrix recurrence:
\begin{equation} \label{eq:fibCmatrix}
\begin{pmatrix}
    F_{n} & F_{n+1} \\F_{n+1} & F_{n+2}
\end{pmatrix} =
\begin{pmatrix}
0 & 1\\
1 & x
\end{pmatrix}
\begin{pmatrix}
    F_{n-1} & F_n \\F_{n} & F_{n+1}
\end{pmatrix}.
\end{equation}
This shows that $F_n(x)$ is the top-right entry of the matrix $C(x)^n$, where
$C(x)$ is the \(2 \times 2\) companion matrix $C(x) = (\begin{smallmatrix} 0 &
1 \\ 1 & x \end{smallmatrix})$. One can again compute $C(x)^N$ using binary
powering, whose costliest step is the multiplication of two polynomial matrices
of degree about $N/2$. This yields $F_N(x)$ in
complexity $O(\timepm{N})$, where $\timepm{N}$ denotes the
cost of polynomial multiplication in degree at most~$N$.

Using FFT-based polynomial multiplication~\cite{CaKa91}, 
this amounts to a number of
operations in the base field~$\K$ which is quasi-linear in~\(N\). 
Not only does this compare favorably to the complexity
$O(N^2)$ of the direct iterative algorithm, but this is even quasi-optimal
(i.e., optimal up to logarithmic factors) with respect to the arithmetic size
$\Theta(N)$ of the output polynomial $F_N(x)$.

In this context, the idea also generalizes to any C-finite sequence
$(u_n(x))_{n\geq 0}$ of polynomials in $\K[x]$, which we will call
\emph{polynomial C-finite sequences}. Indeed, one can encode any recurrence
of arbitrary (but independent of $n$) order $r\geq 1$ and coefficients in
$\K[x]$ into a polynomial $r\times r$ matrix recurrence of order 1, and the
$N$th term of the sequence, $u_N(x)$, can be computed as an element in the
$N$th power of an \(r \times r\) polynomial matrix multiplied by the polynomial
vector of initial values.  Conversely, computing the $N$th power of any
polynomial matrix can be reduced to computing terms in polynomial C-finite
sequences (see the introduction of \cref{sec:bivmodpow_polmatpow}).  Binary
powering allows to solve both problems in $O(\timepm{N})$
operations, and in $O(N^2 \log(N))$ bit operations if $\K=\mathbb{Q}$,
considering both the recurrence order (or the matrix size)~\(r\) and the
recurrence degree (or the matrix degree)~\(d\) as constant parameters, i.e.,
\(r, d \in \bigO{1}\). The main question addressed in this article is:

\smallskip
\hfill \emph{
Can one achieve a better complexity for these tasks?}
\hfill { }
\smallskip

As far as scalar C-finite sequences are concerned, the algebraic
complexity $O(\log(N))$ seems very difficult (if not impossible) to beat,
but it is perhaps not impossible to improve the bit complexity $O(N
\log(N))$ towards $O(N)$. While we do not achieve this, our results
provide polynomial analogues for this type of improvement. As frequently
noticed in computer algebra, polynomials are ``computationally easier'' to
deal with than integers. In our case, philosophically, this comes from the
fact that we can benefit from an additional operation on polynomials:
differentiation.
This possibly cryptic remark will hopefully become clear throughout
\cref{sec:fib_pol}. There, using Fibonacci polynomials as
a test bench, we argue why it is indeed legitimate to hope for algorithms of
complexity $O(N)$ for computing the $N$th term of a polynomial C-finite
sequence.

\myparagraph{Main result}
Recall that a \emph{C-finite sequence} is a sequence $(u_n)_{n\geq0}$ of
elements \(u_n\) in some ring~$R$ which satisfies a recurrence equation
\begin{equation}
  \label{eq:def_polCrec}
  u_{n+r} = c_{r-1} u_{n+r-1} + \cdots + c_0 u_{n} \quad \text{for all} \; n\geq 0,
\end{equation}
for $c_{0},\dots,c_{r-1} \in R$. In this work we consider \emph{polynomial} C-finite sequences,
i.e., the case $R = \K[x]$ for some (effective) field~$\K$ of characteristic zero;
thus $u_n = u_n(x) \in \K[x]$. The customary data structure for representing such a sequence consists of the
polynomials \(c_0(x),\dots,c_{r-1}(x)\) defining the recurrence and the $r$ initial conditions
$u_0(x),\dots,u_{r-1}(x) \in \K[x]$. The \emph{order} of the recurrence~\eqref{eq:def_polCrec} is \(r\) while
its \emph{degree} is the maximum of the degrees of the \(c_i\)'s.

\begin{theorem}
  \label{thm:main}
  Let $\K$ be an effective field of characteristic~0.
  Let \(d\) and \(r\) be fixed positive integers. For each of the following
  problems, there exists an algorithm solving it in
  $O(N)$ operations $(\pm, \times, \div)$ in~$\K$:
  {\renewcommand\labelitemi{}
  \begin{itemize}
    \setlength{\itemindent}{-0.6cm}
    \item \pbseqterm: Given a polynomial C-finite sequence $(u_n(x))_{n \geq 0}$ of order and degree at most $r$ and $d$, compute the $N$th term $u_N(x)$.
    \item \pbbivmodpow: Given polynomials $Q(x,y)$ and $P(x,y)$ in $\K[x,y]$ of degrees in \(y\) and \(x\) at most $r$ and $d$, with $P(x,y)$ monic in $y$, compute $Q(x,y)^N \bmod P(x,y)$.
    \item \pbpolmatpow: Given a square polynomial matrix $M(x)$ over $\K[x]$ of size and degree at most $r$ and $d$, compute $M(x)^N$.
  \end{itemize}
  }
\end{theorem}

Our algorithms for these problems make essential use of divisions in~$\K$. We
do not know if the complexity $O(N)$ can be achieved using only the operations
$(+, -, \times)$ in~$\K$.

\myparagraph{Previous work}
As already mentioned, the classical way of computing the $N$th term of a given
C-finite sequence uses binary powering of the companion matrix, see
e.g.~\cite{MiBr66}. Fiduccia's algorithm \cite{Fiduccia85} utilizes binary powering
in a polynomial quotient ring and improves the complexity with respect to $r$
(but not with respect to~$N$). The fastest known algorithm~\cite{BoMo21} beats
Fiduccia's by a constant factor. In the polynomial C-finite case and assuming
\(r, d \in \bigO{1}\), all these algorithms have a complexity in
$O(\timepm{N})$.

Beyond this classical approach, the previous work on the aforementioned
problems consists of two distinct directions. The special case of
Chebyshev polynomials of the second kind
$U_n(x) = (-i)^{n} F_{n+1}(2ix)$ (with $F_n(x)$ the
$n$th Fibonacci polynomial and $i$ the imaginary unit) was considered
in~\cite{Koepf99} (and later in \cite{Czirbusz12}). These references present various methods for the computation of
the Chebyshev polynomials (of the first and second
kind)
with complexity ranging from $O(N)$ to $O(N^3)$.
The results in~\cite{Koepf99,Czirbusz12} exploit the particular structure
of these polynomials; except for possibly other families of classical
orthogonal polynomials,
for which explicit (hypergeometric) formulas exist, the methods
in~\cite{Koepf99,Czirbusz12} do not admit obvious generalizations.

An idea closely connected to a fundamental building block of our
algorithms is explained in~\cite[Pbm.\,4]{FlSa97}.
%
%
There, Flajolet and Salvy exploit the fact that, given a polynomial
$P(x)$ in $\K[x]$, the coefficient sequence of the $n$th power
$P(x)^n$ satisfies a linear recurrence of order independent of $n$,
and with
coefficients in $\K[x,n]$ of degree independent of~$n$;
this recurrence
allows them to compute (a selected coefficient of)
$P(x)^N$ more efficiently than by binary powering.
This idea has been applied in~\cite[\S8]{BoGaSc07}
to count points on
hyperelliptic curves over finite fields,
with applications to cryptography.
The technique also yields a general solution to
\pbseqterm{} when $r=1$.

\myparagraph{Outline}
The following observation generalizes that in~\cite{FlSa97}: the coefficient
sequence of the $n$th power of any \emph{algebraic function} satisfies a
recurrence of order and degree independent of~$n$. From this, in
\cref{sec:seq_term}, we give algorithms for \pbseqterm{} with cost
\(\bigO{N}\).

To complete the proof of \cref{thm:main}, we design reductions between the
three problems. Obviously $\pbpolmatpow \Rightarrow \pbseqterm$, i.e., any
algorithm for $\pbpolmatpow$ with cost $O(N)$ induces one for $\pbseqterm$ with
cost $O(N)$ as well. Indeed, the $N$th term of a polynomial C-finite
sequence is equal to an entry of the product of the vector of initial values
and the \(N\)th power of a
companion matrix, and this polynomial
vector-matrix multiplication costs $O(N)$. Conversely, it also holds that
$\pbseqterm \Rightarrow \pbpolmatpow$. One natural way to see this is to
consider \(r^2\) sequences corresponding to each entry of \(M(x)^n\), with
recurrence given by the characteristic polynomial of \(M(x)\); see the
introduction of \cref{sec:bivmodpow_polmatpow}. In
\cref{sec:modpow_to_matpow}, we give a more efficient algorithm for this
reduction, based on an algorithm for $\pbseqterm \Rightarrow \pbbivmodpow$
described in \cref{sec:seqterm_to_modpow}.

\myparagraph{Basics of complexity and D-finite functions}
Hereafter, $\K$ denotes an effective field of characteristic zero. We analyze
the performance of algorithms in the algebraic complexity model, meaning that
arithmetic operations $(\pm, \times, \div)$ in the base field $\K$ are counted
at unit cost. As before, $\timepm{N}$ stands for the complexity of multiplying
two polynomials in $\K[x]$ of degree at most $N$. With FFT-based multiplication
$\timepm{N} \in O(N \log(N) \log\log(N))$ \cite{CaKa91}, improved to $O(N
\log(N))$ if $\K$ contains suitable roots of unity~\cite{CoTu65} or if $\K$ 
is a finite
field \cite{HaHo22}.  A power series $f(x) \in \K\ps{x}$ is said to be
\emph{D-finite} if it satisfies a linear differential equation (LDE) of the
form
\begin{equation} \label{eq:dfin}
q_\ell(x) f^{(\ell)}(x) + \cdots + q_0(x) f(x) = 0,
\end{equation}
for some $q_0(x),\dots,q_\ell(x) \in \K[x]$ with $q_\ell(x) \neq 0$. Equivalently, writing $f(x) = \sum_{k \geq 0} f_k x^k$, the sequence $(f_k)_{k \geq 0}$ is \emph{P-finite} (or, \emph{P-recursive}), i.e., it satisfies a linear recurrence equation (LRE)
\[
p_s(k) f_{k+s}  + \cdots + p_0(k) f_k = 0 \quad \text{for all} \; k \geq 0,
\]
with polynomial coefficients $p_0(x),\dots,p_s(x) \in \K[x]$, and $p_s \neq 0$.
Note that $s$ and $\ell$ may differ in general, but $s \leq \ell + \max_i (\deg q_i(x))$. It also holds that $\max_i(\deg p_i(x)) \leq \ell$.

It is often useful to write (\ref{eq:dfin}) as $Lf(x)=0$, where
\[
L = q_\ell(x) \partial_x^\ell + \cdots + q_0(x)
\]
is an element in the noncommutative Weyl algebra $\K[x]\langle \partial_x
\rangle$ of linear differential operators with
multiplication governed by the Leibniz rule
$\partial_x x = x \partial_x +1$.
The \emph{order} $\ell$ of the differential operator~$L$
is the highest power of $\partial_x$ occurring in~$L$,
and the \emph{degree}
of $L$ is the highest power of~$x$ occurring in~$L$. 
We recall that a \emph{least common left
multiple} (LCLM) of two differential operators $L_1, L_2 \in \K[x]\langle \partial_x
\rangle$ is a differential operator $L \in \K[x]\langle \partial_x
\rangle$ 
of minimal order such that 
there exist  $A,B \in
\K(x)\langle \partial_x \rangle$ with $L = AL_1 = BL_2$.
LCLMs can be computed efficiently~\cite{BoChLiSa12}.

\section{The case of Fibonacci polynomials} \label{sec:fib_pol}

Before solving the first part (\pbseqterm) of \cref{thm:main} in general, we propose in this
section three different approaches that can be used to compute the $N$th Fibonacci polynomial
$F_N(x)$ in complexity $O(N)$. Two of these methods have the advantage that they
generalize to the case of arbitrary C-finite sequences.

The starting point of all that follows is the observation that the generating function $F(x,y)
\coloneqq \sum_{n \geq 0} F_n(x) y^n \in \K[x]\ps{y}$ of the sequence $(F_n(x))_{n\geq 0}$ is
\emph{rational}, and equal to $y/(1-xy-y^2)$. 

\subsection{First method via a closed-form expression}
\label{sec:fibclosedform}

By using the partial fraction decomposition
\[
\frac{y}{1-xy-y^2} = \frac{1}{\varphi_+(x) - \varphi_-(x)} \cdot \left(  \frac{1}{1 - \varphi_+(x)y} - \frac{1}{1 - \varphi_-(x)y}  \right)
\]
where $\varphi_{\pm}(x)={(x \pm {\sqrt  {x^{2}+4}})}/{2}$ are the roots of
$\varphi^2 - x\varphi - 1=0$,
and by applying the geometric series, we get the closed-form expression
\begin{equation}\label{eq:Fn-start}
F_n(x) = \frac{\varphi_+(x)^n - \varphi_-(x)^n}{\varphi_+(x) - \varphi_-(x)} \quad \text{for all } n\geq 0.
\end{equation}
Now, using the binomial formula twice, we obtain the formula
\begin{equation}\label{eq:Fn-2}
F_n(x) = \frac{1}{2^{n-1}} \cdot \sum_{\ell \geq 0}   4^\ell
\left(\sum_{k \geq 0}  \binom{n}{2k+1} \binom{k}{\ell} \right) x^{n-2\ell-1}.
\end{equation}
The identity~\cite[3.121]{Gould72} implies a ``magic'' simplification:
\begin{equation}\label{eq:hyperg}
\sum_{k \geq 0}  \binom{n}{2k+1} \binom{k}{\ell} = 2^{n-1-2\ell} \binom{n-\ell-1}{\ell}.
\end{equation}
In conclusion,  from~(\ref{eq:Fn-2}) and (\ref{eq:hyperg}) it follows that
\begin{equation}\label{eq:Fn-final}
F_n(x) =  \sum_{\ell \geq 0}  \binom{n-\ell-1}{\ell}  x^{n-2\ell-1}
.
\end{equation}
With this expression at hand, it becomes transparent that one can compute $F_N(x)$ efficiently. Indeed, by writing $F_N(x) = \sum_{k = 0}^{N-1} f_k x^k$, it follows from (\ref{eq:Fn-final}) that $\seq{f}$ satisfies the recurrence relation
\begin{equation}\label{eq-recFn}
f_{k+2} = \frac{(N + k + 1)(N - k - 1)}{4(k+1)(k+2)} f_k \quad \text{ for all } k \geq 0.
\end{equation}
Moreover, (\ref{eq:Fn-final}) also gives $(f_0,f_1) = (1,0)$ for odd $N$ and otherwise $(f_0,f_1) =
(0,N/2)$. With these initial conditions, it is now clear that
$F_N(x)$ can be computed in $O(N)$ by unrolling the recurrence~\eqref{eq-recFn}.

As mentioned in the introduction, the analogue of
formula~\eqref{eq:Fn-final} for the case of Chebyshev polynomials of the
first kind $T_n(x)$ was already exploited in~\cite[\S1.9]{Koepf99}. The
disadvantage of this approach is that for general polynomial C-finite
sequences there is no hope for a closed-form expression like
(\ref{eq:Fn-final}).

\subsection{Second method via algebraic substitution} \label{sec:fibalgsubs}

There is another method for computing $F_N(x)$ in $O(N)$, which has the advantage that it generalizes to any C-finite sequence, as we will show in \cref{sec:method_algeqtodiffeq}.
The crucial remark (\cref{lem:poweralgeqdiffeq}) is that since
$\varphi_{\pm}(x)$ is algebraic, $\varphi_{\pm}(x)^n$ satisfies a ``small'' LDE, of order and degree independent of $n$. The same holds for
$1/(\varphi_{+}(x)-\varphi_{-}(x))$, therefore
for $F_n(x)$ as well. More precisely,
$\varphi_{\pm}(x)^n$ satisfies the LDE
\[
  (x^2+4) y''(x) + x y'(x) - n^2 y(x) = 0,
\]
and $1/(\varphi_{+}(x)-\varphi_{-}(x)) = (x^{2}+4)^{-1/2}$ satisfies the LDE
\[
  (x^2+4) y'(x) + x y(x) = 0.
\]
Using~\eqref{eq:Fn-start}, it then follows that
the polynomial $F_n(x)$ satisfies
\begin{equation} \label{eq:diffeqfib}
  (x^2+4) y''(x) + 3x y'(x) + (1-n^2) y(x) = 0.
\end{equation}
Writing $F_N(x) = \sum_{k = 0}^{N-1} f_k x^k$, plugging into
(\ref{eq:diffeqfib}) for $n=N$ and extracting the $(k+2)$nd coefficient,
it now follows that the sequence $(f_k)_{k\geq0}$ satisfies
recurrence~\eqref{eq-recFn}.
The initial conditions $f_0, f_1$ are given by $F_N(x) \bmod x^2$ which
can be found in complexity $O(\log(N))$ by computing the $N$th
power of the companion matrix~(\ref{eq:fibCmatrix}) in $\K[x]/(x^2)$ by
binary powering and reducing mod $x^2$ in each step. As before, unrolling
recurrence \eqref{eq-recFn} with these initial terms provides a way to
compute $F_N(x)$ in complexity $O(N)$.

\subsection{Third method via Creative Telescoping}\label{sec:CTfib}

Writing $F(x,y) = y/(1-xy-y^2)$ we are interested in a differential equation for the coefficient of $y^N$ in $F(x,y)$. 
By Cauchy's integral formula, we have for sufficiently small $\epsilon>0$:
\[
  F_N(x) = [y^N] F(x,y) = \frac{1}{2\pi i} \oint_{|y|=\epsilon} \frac{y}{(1-xy-y^2)y^{N+1}} \mathrm{d}y.
\]
Then the method of creative telescoping~\cite{AlZe90} can be used to find an LDE for the integral above. For example, the command
\begin{verbatim}
    DEtools[Zeilberger](1/(1-x*y-y^2)/y^n, x, y, Dx);
\end{verbatim}
in Maple immediately finds that
\begin{align*}
  \left((x^2+4) \partial_x^2+3x\partial_x + 1 - n^2\right) \frac{F(x,y)}{y^{n+1}} = \partial_y \left(\frac{F(x,y)}{y^{n}} C(x,y) \right),
\end{align*}
where $C(x,y) = (n+1 - nxy - (n -1) y^{2})/(1-xy-y^2)$. By Cauchy's integral theorem, the contour integral of the right-hand side vanishes, and (\ref{eq:diffeqfib}) follows.
Then one can conclude in the same way as in the previous method and compute $F_N(x)$ in complexity $O(N)$.

\subsection{Comments on the three approaches} It is natural to ask
ourselves what in these approaches was just luck, what was truly specific
to the particular example of the Fibonacci polynomials, and what can be
extended to the general case.

It is clear that the key for computing $F_N(x)$ in complexity $O(N)$ is
the existence of the recurrence (\ref{eq-recFn}) (or equivalently 
the LDE~\eqref{eq:diffeqfib}). Even though there is no hope for a closed-form
solution in general, we shall prove that such a recurrence always exists
for polynomial C-finite sequences. We should, however, definitely be
careful and avoid proving tautologic statements. Since $u_N(x)$ is a
polynomial, it does satisfy the first-order LDE $u_N(x) y'(x) -
u_N'(x) y(x)=0$, but this one is trivial for our purposes. Indeed,
converting this differential equation into a recurrence satisfied by the
sequence of coefficients of $u_N(x)$ yields a recurrence of order $\deg(u_N)$,
which is obviously useless for computing the coefficients of $u_N$. Rather, we
would like to find an LRE/LDE whose order and degree are
independent of $N$. This is the purpose of the next section. Specifically,
in \S\ref{sec:method_algeqtodiffeq} we explain how it can be computed by
algebraic substitution (generalizing~\S\ref{sec:fibalgsubs}) and in
\S\ref{sec:CT} we show that it can also be found via creative telescoping (generalizing \S\ref{sec:CTfib}).

\section{Polynomial C-finite sequences}
\label{sec:seq_term}

Recall that a polynomial C-finite sequence $(u_n(x))_{n\geq0}$ is a sequence of polynomials $u_n(x) \in \K[x]$ that satisfies a recurrence
\begin{equation} \label{eq:def_polCrec2}
   u_{n+r}(x) = c_{r-1}(x) u_{n+r-1}(x) + \cdots + c_0(x) u_{n}(x) ,
\end{equation}
of some order $r \in \NN$, with coefficients $c_0(x),\dots,c_{r-1}(x) \in \K[x]$. The \emph{degree} of
(\ref{eq:def_polCrec2}) is $d = \max_i (\deg c_i(x))$. The sequence $(u_n(x))_{n\geq0}$ is defined uniquely by \eqref{eq:def_polCrec2} if $r$ initial terms $u_0(x),\dots,u_{r-1}(x)$ are
prescribed. The characteristic polynomial of (\ref{eq:def_polCrec2}) is defined as
\[
\chi(y) =
y^r  -  c_{r-1}(x)y^{r-1}  - \cdots -  c_{1}(x) y - c_0(x) \in \K[x,y].
\]
The generating function $U(x,y) \coloneqq \sum_{n \geq 0} u_n(x)y^n$ is rational:
\begin{equation} \label{eq:genfunc}
U(x,y) = \frac{v_0(x) + \cdots + v_{r-1}(x)y^{r-1}}{y^r \chi(1/y)},
\end{equation}
with $v_k(x) \coloneqq
u_k(x) - c_{r-1}(x) u_{k-1}(x) - \cdots - c_{r-k}(x)u_0(x)$.

Let $a_1(x),\dots,a_k(x) \in \closure{\ffield}$ be the roots of $\chi(y)$, and $m_1,\dots,m_k$ be their multiplicities.
By partial fraction decomposition and geometric series, any sequence $(u_n(x))_{n\geq0}$ satisfying (\ref{eq:def_polCrec2}) has the form
\begin{equation} \label{eq:polCrecSol}
u_n(x) = q_1(n,x) a_1(x)^n + \cdots + q_k(n,x) a_k(x)^n,
\end{equation}
where $k \leq r$ and each $q_i(n,x) \in \field(a_1(x),\dots,a_n(x))[n]$ is a polynomial in $n$ of degree at most $m_i-1$, for $i = 1,\dots,k$.


\subsection{Computing \texorpdfstring{$u_N(x)$}{u\_n(x)} in \texorpdfstring{$O(N)$}{O(N)}}
\label{sec:method_algeqtodiffeq}

By generalizing the ideas of \cref{sec:fibalgsubs}, it is not difficult to prove that the $n$th term of a polynomial C-finite sequence $(u_n(x))_{n\geq0}$ satisfies an LDE whose order and degree are independent of $n$, and consequently, that there exists a linear recurrence relation for the coefficient sequence of $u_n(x)$ whose order (say $s$) and degree are again independent of $n$. Then, for a given $N \in \NN$, first computing initial terms by binary powering of the companion matrix in $\K[x]/(x^s)$ and then unrolling this recurrence for $n= N$, we achieve a complexity~$O(N)$ for the computation of $u_N(x)$.

\begin{theorem} \label{thm:Ln}
Let $(u_n(x))_{n\geq0}$ be a polynomial C-finite sequence. Then there exists $L_n \in \K[n,x]\langle \partial_x \rangle$
with order and degree independent of $n$,
and
such that $L_n (u_n(x)) = 0$. Consequently,
writing $u_n(x) = \sum_{k \geq 0} c_{n,k} x^k$,
there exist, for some $s \in \NN$ independent of $n$,
polynomials $p_0(n,x),\dots,p_s(n,x) \in \K[n,x]$
of degrees independent of~$n$,
and
such that the sequence $(c_{n,k})_{k \geq 0}$ satisfies the recurrence
\begin{equation} \label{eq:recuN}
p_s(n,k) c_{n,k+s} + \cdots + p_0(n,k) c_{n,k} = 0, \quad k\geq0.
\end{equation}
\end{theorem}

In the theorem above it is crucial that neither the order nor the degree of $L_n$ depend on $n$. Since each $u_n(x)$ is a polynomial, it is a tautology to say that it satisfies \emph{some} LDE: one may simply take $L = \partial_x^\alpha$, where $\alpha>\deg(u_n(x))$ or $L = u_n(x) \partial_x - u_n'(x)$. However, it is a nontrivial fact that $u_n(x)$ satisfies an LDE of the form
\[
p_{\ell}(n, x) u_n^{(\ell)}(x) + \dots + p_0(n, x) u_n(x) = 0
\]
for some $p_i(n,x) \in \K[n,x]$ with $\ell$ and $\deg_x p_i$ independent of $n$.

The most direct proof of \cref{thm:Ln} uses the explicit
expression~\eqref{eq:polCrecSol} for $u_n(x)$ and the following classical fact
about algebraic substitution into D-finite functions.  Recall that a function
$a(x)$ is called \emph{algebraic} over $\K(x)$ if it satisfies a nontrivial
polynomial relation $P(x,a(x))=0$ for some $P(x,y) \in \K[x,y]$. Size and
complexity bounds on differential equations for algebraic functions, and more
generally on algebraic substitution, are given in~\cite{BCLSS07,KaPo17}.

\begin{lemma} \label{lem:poweralgeqdiffeq}
    Let $a(x)$ be an algebraic function over $\K(x)$ and let $g(x)$ be D-finite. Then $f(x) = g(a(x))$ is D-finite. In particular, $a(x)^n$ satisfies an LDE of order and degree independent of $n$.
\end{lemma}
\begin{proof}
    The first part is a classical result, see for example~\cite[Thm.~2.7]{Stanley80}. In the proof one shows that the vector space spanned over $\K(x)$ by $(f^{(i)}(x))_{i\geq0}$ is finite-dimensional over $\K(x,a(x))$ which is itself finite-dimensional over $\K(x)$. For the second part, it is enough to set $g(x) = x^n$ which satisfies $x g'(x) = n g(x)$.
\end{proof}

\begin{example}
    Like in \cref{sec:fib_pol} let $\varphi_{\pm}(x) = (x\pm\sqrt{x^2+4})/2$ be the roots of $y(x)^2+xy(x)-1=0$. Then $\varphi_{\pm}(x)^n$ satisfy the LDE
    \[
    (x^2+4) y''(x) + x y'(x) - n^2 y(x) = 0.
    \]
\end{example}

\begin{proof}[Proof of \cref{thm:Ln}]
Write $u_n(x)$ as in (\ref{eq:polCrecSol}). By \cref{lem:poweralgeqdiffeq}, each $a_i(x)^n$ satisfies an LDE of order and degree independent of $n$, hence the same holds for $q_i(n,x)a_i(x)^n$, and finally for $u_n(x)$. It follows that the coefficient sequence of $u_n(x)$ is P-finite with order and degree independent of $n$.
\end{proof}

Since all steps in the proofs above are effective and independent of $N$, this
leads to \cref{alg:seqtermas}. Its
\crefrange{step:seqtermas:charpoly}{step:seqtermas:recurrence} can be seen as
``precomputations'' since they do not depend on $N$. As already mentioned,
\cref{step:seqtermas:init_values} has complexity $O(\log(N))$ and
\cref{step:seqtermas:unroll} has complexity $O(N)$. Thus, \cref{alg:seqtermas}
solves \pbseqterm{} in complexity $O(N)$, up to a potential issue during the
unrolling at \cref{step:seqtermas:unroll} of the recurrence from
\cref{step:seqtermas:recurrence}. Indeed, this unrolling may be impossible for
some values $k$, namely those for which $p_s(N,k)$ vanishes. We will explain
how to overcome this problem in~\cref{sec:singular}.

For practical applications, however, computing the polynomials $q_i(x,n)$ in
\cref{step:seqtermas:minpoly} as well as the LCLM in \cref{step:seqtermas:lclm}
is algorithmically somewhat cumbersome. Thus, generalizing the approach in
\cref{sec:CTfib}, we now propose a variant of \cref{alg:seqtermas} which
replaces \crefrange{step:seqtermas:charpoly}{step:seqtermas:lclm} by an
algorithm based on creative telescoping.

\begin{algorithm}[t]
	\caption{SeqTermAS\(((u_n)_n, N)\)}
	\label{alg:seqtermas}
  \begin{algorithmic}[1]
	  \Require{A polynomial C-finite sequence $(u_n(x))_{n\geq0}$ given by (\ref{eq:def_polCrec2}) with initial conditions, and $N \in \NN$.}
    \Ensure{The polynomial $u_N(x)$.}

    \State $d \gets \deg_x(\chi(y))$ and $\delta \gets \max_{i=0,\dots,r-1}(\deg_x u_i(x))$

    \State\label{step:seqtermas:charpoly}%
    \(\chi(y) \gets \) the characteristic polynomial of $(u_n(x))_{n\geq0}$

    \State \label{step:seqtermas:roots}%
    $a_1(x), \ldots, a_k(x) \gets$ the roots of \(\chi(y)\)
    
    \State \label{step:seqtermas:minpoly}%
    Compute minimal polynomials for $q_1(x,n),\dots,q_k(x,n)$ $ \in \K(a_1(x),\dots,a_k(x))[n]$ such that (\ref{eq:polCrecSol}) holds.

    \State \label{step:seqtermas:ODE}%
    For each $i$ deduce an LDE $L_{i,n} \in \K[n,x]\langle \partial_x \rangle$ with order and degree independent of $n$ such that $L_{i,n}(q_i(x,n) a_i(x)^n) = 0$

    \State \label{step:seqtermas:lclm}%
    $L_n \gets \mathrm{LCLM}(L_{1,n},\dots,L_{k,n}) \in \K[n,x]\langle \partial_x \rangle$

    \State \label{step:seqtermas:recurrence}%
    Compute a recurrence
    \(
        p_s(n,k) c_{n,k+s} + \dots + p_0(n,k)c_{n,k} = 0
    \)
    satisfied by any solution $f_n(x) = \sum_{k \geq 0} c_{n,k} x^k$ of $L_ny=0$

    \State \label{step:seqtermas:init_values}%
    Using binary powering of the companion matrix of the initial recurrence mod $x^s$, compute the values $c_{N,0},\dots,c_{N,s-1}$

    \State \label{step:seqtermas:unroll}%
    Unroll the recurrence from \cref{step:seqtermas:recurrence} for $n=N$ and with initial terms from \cref{step:seqtermas:init_values}
    (see also \cref{sec:singular} for further details)

    \State\label{step:seqtermas:return}%
    \Return $\sum_{k = 0}^{Nd + \delta} c_{N,k} x^k$
  \end{algorithmic}
\end{algorithm}

\subsection{Computing \texorpdfstring{$L_n$}{L\_n} with Creative Telescoping} \label{sec:CT}
Let $U(x,y) = \sum_{n\geq0} u_n(x)y^n \in \K[x]\ps{y}$ be the generating function~\eqref{eq:genfunc} of $(u_n(x))_{n\geq0}$. The sequence is C-finite, so $U(x,y)$ is a rational function. Moreover, the Cauchy integral formula implies \begin{equation} \label{eq:Cauchy}
u_{n}(x) = \frac{1}{2\pi i} \oint_{|y| = \epsilon} \frac{U(x,y)}{y^{n+1}} \mathrm{d}y.
\end{equation}
A \emph{telescoper} of $U(x,y)/y^{n+1}$ is a differential operator
\[
L = p_k(x) \partial_x^k + \cdots + p_0(x) \in \K[x]\langle \partial_x \rangle,
\]
such that $L$ applied to $U(x,y)/y^{n+1}$ is $\partial_y(C(x,y))$ for some rational function $C(x,y)$ called the \emph{certificate}. By the Cauchy integral theorem, $\oint_{|y|=\epsilon} \partial_y(C(x,y)) \mathrm{d}y = 0$, and it follows that $L u_n(x) = 0$, i.e., $L$ yields a differential equation for $u_n(x)$. In this section we will prove that for $U(x,y)/y^{n+1}$ there exists a telescoper $L_n \in \K[n,x] \langle \partial_x \rangle$ whose order and degree do not depend on $n$. Our proof relies on reduction-based creative telescoping and repeatedly uses \emph{Hermite reduction} algorithms~\cite{BoChChLi10,BoChChLi13}.

We now introduce the necessary definitions and recall the
Hermite reduction method. For a more detailed introduction, a
full complexity analysis, and applications of reduction-based creative
telescoping to integration of bivariate rational functions, we refer
to~\cite{BoChChLi10}. Let $\LL = \K(x)$. For a polynomial $Q(y) \in
\LL[y]$, let $Q = Q_1 Q_2^{2} \cdots Q_k^{k}$ be its squarefree
factorization and let $Q^* = Q_1 \cdots Q_k$ denote the squarefree part of
$Q$. We set $Q^- \coloneqq Q/Q^*$. Recall that, given $P,Q \in \LL[y]$,
the \emph{Hermite reduction} algorithm computes two polynomials $A, a \in
\LL[y]$ with $\deg_y a < \deg_y Q^*$ such that
\[
\frac{P}{Q} = \partial_y \left( \frac{A}{Q^-} \right) + \frac{a}{Q^*}.
\]
Given a bivariate rational function $H(x,y) = P(y)/Q(y) \in \K(x,y)$, one may compute the Hermite reduction $(A_i,a_i)$ of $\partial_x^i H$ for $i=0,1,\dots$. Since $\deg_y a_i$ is uniformly bounded by $d^* = \deg_y Q^*$ for each $i$, the $d^*+1$ functions $\{a_i(x,y) \colon 0 \leq i \leq d^*\}$ will be linearly dependent over $\K(x)$. Hence one can find $q_0(x),\cdots,q_{d^*}(x) \in \K(x)$ not all zero, such that $\sum_{i=0}^{d^*} q_i(x)a_i(x) = 0$. It follows then that $L = \sum_{i=0}^{d^*} q_i(x) \partial_x^i$ is a telescoper for $H$.

This procedure cannot be directly applied to $U(x,y)/y^{n+1}$ if $n$ is an indeterminate. At the same time, if $n=N \in \NN$ is fixed, it is \textit{a priori} not obvious that $\deg_x q_i(x)$ will be independent of $N$. Moreover, the complexity of the algorithm will depend on $N$, which we want to avoid. As we will now explain, to achieve this, one should see $U(x,y)/y^{n+1}$ not as a rational function in $x$ and $y$ with potentially large degree in the numerator, but as a hyperexponential function with the parameter $n$ appearing solely as a coefficient in the logarithmic derivative. Recall that $H(x,y)$ is called \emph{hyperexponential} if both $\partial_x H/H$ and $\partial_y H/H$ belong to $\K(x,y)$.

For hyperexponential functions, the Almkvist-Zeilberger
algorithm~\cite{AlZe90} was the first practical method to find telescopers
and certificates. Indeed, as we mentioned in \cref{sec:CTfib}, the command
\begin{verbatim}
    DEtools[Zeilberger](1/(1-x*y-y^2)/y^n, x, y, Dx);
\end{verbatim}
in Maple immediately finds the differential equation for the $n$th Fibonacci polynomial for a variable $n$. Note that if $n$ is specialized to an integer $N$ before the execution of the command above, the implemented algorithm becomes slower as $N$ grows.

It is, however, not clear that the Almkvist-Zeilberger algorithm applied to $U(x,y)/y^{n+1}$ will always find a telescoper whose degree and order are independent of $n$, even though we know from Section~\ref{sec:method_algeqtodiffeq} that an LDE with this property exists. Therefore, to have a complete algorithm based on creative telescoping, we will invoke the reduction-based method for hyperexponential functions first introduced and analyzed in~\cite{BoChChLi13}. Using the implementation of the latter work, the command in Maple
\begin{verbatim}
    HermiteTelescoping(1/(1-x*y-y^2)/y^n, x, y, Dx);
\end{verbatim}
also immediately finds the correct LDE for $F_n(x)$. The practical advantage for our purpose of using the reduction-based algorithm in comparison to the Almkvist-Zeilberger method is shown in \cref{sec:experiments} (\cref{tab:CT}). The theoretical advantage comes from the following lemma, which guarantees that the algorithm will find a telescoper for $U(x,y)/y^{n+1}$, and consequently an LDE for $u_n(x)$, whose order and degree do not depend on $n$.

\begin{lemma} \label{lem:hermiteCT}
    Let $P(y) \in \LL[n,y]$ and $Q(y) \in \LL[y]$ with $Q(0) \neq 0$. Set $d_n \coloneqq \deg_n P(y)$, $d^* \coloneqq \deg_y Q^*(y)$ and let $k$ be the highest pure power in the square free factorization of $Q(y)$. Then there exist $B(n,y) \in \LL(n)[y]$ and $b(n,y) \in \LL[n,y]$ with $\deg_y b(n,y) \leq d^*$ and $\deg_n b(n,y) \leq d_n + k$ such that
    \begin{equation} \label{eq:Reductionwish}
        \frac{P(y)}{Q(y) y^{n+1}} = \partial_y \left( \frac{B(n,y)}{Q^-(y)y^{n}} \right) + \frac{b(n,y)}{Q^*(y) y^{n+1}}.
    \end{equation}
\end{lemma}

\begin{proof}
We are going prove the statement by induction on $d^- \coloneqq \deg_y Q^-$. If $d^- = 0$, then $Q^* = Q$ and the Euclidean division gives $P = P_1 Q + b_1$ with $\deg_y b_1 < d^*$. Moreover,
\[
\frac{P_1(y)}{y^{n+1}} = \partial_y \left( \frac{B_1(n,y)}{y^n} \right),
\]
where $B_1(n,y)$ is $P_1(y)$ with the $k$th coefficient $p_k$ replaced by $p_k/(k-n)$. Setting $B=B_1$ and $b = b_1$ proves the induction basis.

Now assume that $d^->0$ and note that
\[
\partial_y \left( \frac{B(n,y)}{Q^-(y)y^{n}} \right) =  y^{-n} \partial_y \left( \frac{B(n,y)}{Q^-(y)} \right) - y^{-n-1} n  \frac{B(n,y)}{Q^-(y)},
\]
so equation (\ref{eq:Reductionwish}) is equivalent to
\begin{equation} \label{eq:1}
\frac{P(y)}{Q(y) y} = \partial_y \left( \frac{B(n,y)}{Q^-(y)} \right)  - n  \frac{B(n,y)}{Q^-(y) y } + \frac{b(y)}{Q^*(y) y}.
\end{equation}
The Hermite reduction applied to $\frac{P(y)}{Q(y) y}$ yields $A(y), a(y) \in \LL[n,y]$ with $\deg_y a(y) \leq d^*$ and $\deg_n A(y), \deg_n a(y) \leq d_n$ such that
\begin{equation} \label{eq:2}
\frac{P(y)}{Q(y) y} = \partial_y \left( \frac{A(y)}{Q^-(y)} \right) + \frac{a(y)}{Q^*(y)y}.
\end{equation}
Comparing (\ref{eq:1}) and (\ref{eq:2}), we now look at
\[
H(y) \coloneqq \frac{a(y)}{Q^*(y) y} + n\frac{A(y)}{Q^-(y) y}.
\]
The denominator of $H(y)$ is $R(y) y \coloneqq \mathrm{lcm}(Q^*, Q^-) y$. Clearly, $R^* = Q^*$ and $\deg_y R^- < d^-$. The highest pure power in the square free factorization of $R(y)$ is at most $k-1$ and the degree of the numerator in $n$ of $H(y)$ is at most $d_n+1$. Hence, by induction, we may write
\[
\frac{a(y)}{Q^*(y) y} + n\frac{A(y)}{Q^-(y) y} = \partial_y \left( \frac{C(n,y)}{Q^-(y)} \right)  - n  \frac{C(n,y)}{Q^-(y) y } + \frac{c(n,y)}{Q^*(y) y}
\]
with $\deg_y c(n,y) < d^*$ and $\deg_n c(n,y) \leq d_n+k$. Setting $B(n,y) = A(y) + C(n,y)$ and $b(y) = a(y) + c(n,y)$ finishes the proof.
\end{proof}

The proof of \cref{lem:hermiteCT} induces an algorithm for the computation of $B(n,y)$ and $b(n,y)$ given $P(y),Q(y) \in \K[x,y]$ such that (\ref{eq:Reductionwish}) holds, $\deg_y b \leq \deg_y Q^*$ and also $\deg_n b$ bounded in terms of $Q$. It can be seen as a special case of the procedure \textsc{HermiteReduction} in \cite{BoChChLi13}.
The LDE for $u_n(x)$ can now be found as in \cref{alg:telescnthterm}.

\begin{algorithm}[ht]
	\caption{TelescNthTerm\((U(x,y))\)}
	\label{alg:telescnthterm}
  \begin{algorithmic}[1]
    \Require{A rational function $U(x,y) \in \K(x,y) \cap \K\ps{x,y}$.}
    \Ensure{A diff.\ operator in $\K[n,x] \langle \partial_x\rangle$ for $u_n(x) = [y^n] U(x,y)$.}

    \State Write $U(x,y) = P(x,y)/Q(x,y)$ and let $d^* = \deg_y Q^*(x,y)$
    \State For each $i = 0,\dots,d^*$ compute the polynomial $b_i(n,x,y) = b(n,y)$ as in \cref{lem:hermiteCT} applied to $\partial_x^i U(x,y)/y^{n+1}$
    \State \label{step:telescnthterm:linear_relation}%
    Find a linear relation of $\{b_i(n,x,y) \colon 0 \leq i \leq d^*\}$ over $\K(n,x)$, that is polynomials $q_0(n,x),\dots,q_{d^*}(n,x)$ not all zero with $\sum_{i=0}^{d^*} q_i(n,x) b_i(n,x,y) = 0$
    \State Return the differential operator $\sum_{i=0}^{d^*} q_i(n,x) \partial_x^{i} \in \K[n,x]\langle \partial_x \rangle$
  \end{algorithmic}
\end{algorithm}

Note that, as in the usual reduction-based creative telescoping, the linear relation at
\cref{step:telescnthterm:linear_relation} exists because $\deg_y( b_i(n,x,y))$ is uniformly
bounded by $d^*$. Writing $U(x,y) = P(x,y)/Q(x,y)$, the operator $L = \sum_{i=0}^{d^*} q_i(n,x)
\partial_x^{i}$ annihilates $u_n(x) = [y^n] U(x,y)$ since
\begin{align*}
2\pi i \cdot L_n u_n(x) =  L_n \oint
 \frac{U(x,y)}{y^{n+1}} & \mathrm{d}y = \oint  L_n
 \frac{P(x,y)}{Q(x,y)y^{n+1}} \mathrm{d}y \\= \oint \partial_y \frac{\sum_{i=1}^{d^*} q_i(n,x) B_i(n,y)}{Q^-(y) y^n} & \mathrm{d}y +  \oint \frac{\sum_{i=0}^{d^*} q_i(n,x) b_i(n,x,y)}{Q^*(y) y^{n+1}} \mathrm{d}y;
\end{align*}
the first integral vanishes by Cauchy's integral theorem, 
and the second integral vanishes by construction of the $q_i(n,x)$.

\begin{algorithm}[t]
	\caption{SeqTermCT\(((u_n)_n, N)\)}
	\label{alg:seqtermct}
  \begin{algorithmic}[1]
	  \Require{A polynomial C-finite sequence $(u_n(x))_{n\geq0}$ given by (\ref{eq:def_polCrec2}) with initial conditions, and $N \in \NN$.}
    \Ensure{The polynomial $u_N(x)$.}

    \State $U(x,y) \gets$ the rational generating function of $u_n(x)$ in (\ref{eq:genfunc})

    \State $L_n \gets$ $\Call{TelescNthTerm}{U(x,y)}$

    \State\CommentLine{follow \crefrange{step:seqtermas:recurrence}{step:seqtermas:return} of \cref{alg:seqtermas}}



  \end{algorithmic}
\end{algorithm}

This provides a variant for
\crefrange{step:seqtermas:charpoly}{step:seqtermas:lclm} of \cref{alg:seqtermas},
as described in \cref{alg:seqtermct}.
The above-mentioned potential issue with unrolling persists; the next section deals with this problem.

\subsection{The singular case} \label{sec:singular}
In this section we discuss the potential issue of our algorithm that can occur if the sequence for the coefficients of $u_n(x)$ cannot be unrolled due to singularities. We shall first highlight this problem and its solution by means of an example, then in the last paragraph of this section we explain the general strategy.

Consider the polynomial C-finite sequence $u_n(x)$ given by
\[
u_{n + 3}(x) - (x^2 + x + 2)u_{n + 2}(x)  + x(x^2 + 2x + 2)u_{n + 1}(x) - 2x^3u_n(x) = 0,
\]
for all $n\geq0$ with initial conditions $u_0 = 3, u_1 = x^2 + x + 2, u_2 = x^4 + x^2 + 4$. The characteristic polynomial of the defining recurrence is easily computed and turns out to factor completely over $\K[x][y]$:
\[
\chi(x,y) = (y - 2)(y-x)(y-x^2).
\]
With the initial conditions and after a partial fraction decomposition it follows that the generating function of $u_n(x)$ is given by
\[
U(x,y) = \frac{1}{1-2 y}+\frac{1}{1- x^{2} y}+\frac{1}{1 - x y}.
\]
Hence, the solution is $u_n(x) = 2^n + x^n + x^{2n}$ and can be written down in $O(\log(N))$ operations. However, as we shall explain now, the direct application of any of the methods described earlier fails.

According to \cref{thm:Ln}, $u_n(x)$ satisfies an LDE whose degree and order are independent of $n$. Indeed, using creative telescoping one quickly finds an annihilator for $u_n(x) =  \oint U(x,y)/y^{n+1} \mathrm{d}x$:
\[
\bigl( x^2\partial_x^3 - 3x(n - 1)\partial_x^2 + (2n - 1)(n - 1)\partial_x \bigr) u_n(x) = 0.
\]
Converting this LDE to a recurrence for the coefficient sequence of $u_n(x) = \sum_{k\geq0} c_{n,k} x^k$ we find
\begin{equation} \label{eq:rec_singularec}
    (2n-k) (n-k) k c_{n,k} = 0, \quad k\geq0.
\end{equation}
In other words, $c_{n,k} = 0$ for all $k \in \NN$ except $k\in\{0,n,2n\}$. In
order to ``unroll'' this recurrence we need to know $c_{n,0},c_{n,n}$ and
$c_{n,2n}$. However, it is not immediately clear how to compute those terms for
$n=N$ in $O(N)$ operations from the initial input (without using the
explicit solution).

We propose the following easily generalizable solution: consider  $v_n(x) = u_n(x+1)$. Then the LDE for $v_n(x)$ is given by
\[
\bigl( (x+1)^2\partial_x^3 - 3(x+1)(n - 1)\partial_x^2 + (2n - 1)(n - 1)\partial_x \bigr) v_n(x) = 0,
\]
and for the coefficient sequence of $v_n(x) = \sum_{k\geq0} d_{n,k} x^k$ we find
\[
(k + 1)(k + 2) d_{n,k+2} - (k + 1)(3N -2k - 1) d_{n,k+1} + (2n - k)(n - k)d_{n,k}= 0.
\]
Now the leading coefficient of the recurrence is $(k + 1)(k + 2) \neq 0$, so we can easily unroll it after determining the first two terms, by computing them via binary powering of the corresponding companion matrix mod~$x^2$. Having computed $v_N(x)$, it remains to find $u_N(x) = v_N(x-1)$. Note that expanding the polynomial results in an $O(\timepm{N})$ algorithm. However, recall from (\ref{eq:rec_singularec}) that we only need to compute $c_{N,N}$ and $c_{N,2N}$, or, in other words, the coefficients of $x^N$ and $x^{2N}$ in $v_N(x-1)$. For any $i$ it holds that
\begin{equation} \label{eq:dntocn}
c_{N,i} = \sum_{k \geq 0} d_{N,k} \binom{k}{i} (-1)^{k-i},
\end{equation}
and the sum is finite because $v_N(x)$ is a polynomial. Clearly, it can be computed in complexity $O(N)$ for any $i$.

Generally speaking, an issue with unrolling the recurrence for $(c_{n,k})_{k \geq 0}$ occurs if the roots of the leading polynomial are positive integers that depend on $n$. Indeed, roots that are nonintegral clearly do not cause any problems in the unrolling step and if a root is independent of $n$ then we may just compute more initial terms while the complexity of this step stays bounded by $O(\log N)$. Let $S$ be the set of the problematic roots. Note that the size of $S$ is independent of $n$ since $S$ is a subset of the roots of the leading polynomial $p_s(n,x)$ in \eqref{eq:recuN} and $\deg_x p_s(n,x)$ is bounded independently of $n$ by \cref{thm:Ln}. Moreover, $S$ can be nonempty only if the LDE for $u_n(x)$ is singular at 0 (that is, if $q_\ell(x) = q_\ell(n,x)$ in (\ref{eq:dfin}) vanishes at $x=0$). In this case, one can always define $v_n(x) = u_n(x+c)$ for $c \in \K$ a nonsingular point of the LDE ($q_\ell(n,c) \neq 0$). Then the coefficients $d_{n,k}$ of $v_n(x)$ can be computed from $O(1)$ initial conditions via unrolling a recurrence. Using the formula (\ref{eq:dntocn}) (with $-c$ instead $-1$) and the fact that $v_n(x)$ is a polynomial, one can compute the coefficients $c_{N,i}$ for $i \in S$. It is then possible to unroll the recurrence for $(c_{N,k})_{k \geq 0}$ and find $u_N(x)$ in complexity~$O(N)$.


\section{Impact on polynomial matrix power}
\label{sec:bivmodpow_polmatpow}

Here is an algorithm for \pbpolmatpow{} using \pbseqterm. Let \(M(x)\) in
\(\pring_{\le d}^{r \times r}\) and \(p_{i,j,n}(x)\) be the \((i,j)\) entry of
\(M(x)^n\), for \(n \ge 0\) and \(i\) and \(j\) in \(\{1,\ldots,r\}\). The
sequence $(p_{i,j,n}(x))_{n\geq 0}$ is polynomial C-finite, with a
recurrence given by the characteristic polynomial
\[
  \chi_M(x,y)
  \coloneqq \det(y \, I_r - M(x))
  = y^r - c_{r-1}(x) y^{r-1} - \cdots - c_0(x) .
\]
That is, $p_{i,j,n+r}(x) = c_{r-1}(x) p_{i,j,n+r-1}(x) + \cdots + c_0(x)
p_{i,j,n}(x)$ for all $n\geq 0$. Thus, to compute $M(x)^N$, it is enough to
find $\chi_M(x,y)$ (in $O(1)$, i.e.\ independent of $N$), to compute the
polynomials $p_{i,j,n}(x)$ for $1\leq i, j \leq r$ and $0\leq n <r$ (also in
$O(1)$) and to return the entries $p_{i,j,N}(x)$ of $M(x)^N$ via \pbseqterm.
As such, this approach uses \(r^2\) calls to \pbseqterm, with total cost
\(\bigO{N}\).

This section describes an algorithm for \pbpolmatpow{} which uses only \(r\)
such calls, through a direct reduction to \pbbivmodpow{} (see
\cref{sec:modpow_to_matpow}). Our solution for \pbbivmodpow, via \(r\) calls to
\pbseqterm, is presented in \cref{sec:seqterm_to_modpow} and completes the
proof of \cref{thm:main}.

\subsection{Computing bivariate modular powers}
\label{sec:seqterm_to_modpow}

Let $\LL = \K[x]$ and $P, Q \in \LL[y]$. Assume that $P$, seen as a univariate
polynomial in \(y\) of degree \(r\), is monic. For $N \in \NN$, the Euclidean
division in $\LL[y]$ ensures the existence of unique $S, R \in \LL[y]$ such
that $\deg_y(R) < r$ and \(Q^N = S P + R\). The polynomial $R$ is $Q^N \bmod
P$. Assume that \(P\) and \(Q\) are fixed, and let $d \coloneqq \deg_x(P)$
(which is thus in \(O(1)\)). Then, writing $R = \sum_{i=0}^{d-1} r_i(x)y^i$, it
holds that $\deg_x r_i(x) = O(N)$. The efficient computation of $R$ when $Q=y$,
given $P(x,y)$ and $N$, is the first step for proving \pbbivmodpow{} in
\cref{thm:main}.

\begin{algorithm}[b]
	\caption{BivModPow\((P(x,y), Q(x,y), N)\)}
	\label{alg:bivmodpow}
  \begin{algorithmic}[1]
    \Require{$P(x,y),Q(x,y) \in \LL[y]$ with $P(x,y)$ monic in $y$, and $N \in \NN$.}
    \Ensure{$Q(x,y)^N \bmod P(x,y)$.}

    \State $A(x,t) \gets \mathrm{Res}_y(P(x,y),t-Q(x,y))$
    \State $\bar{A}(x,t) \gets t^r A(x,1/t)$, where $r \coloneqq \deg_t A(x,t)$
    \State \InlineFor{$i = N-r+1,\dots,N$}{\(u_{i}(x) \gets [t^{i}] \frac{1}{\bar{A}(x,t)}\) using \pbseqterm}
    \State $u(x,t) \gets u_{N-r+1}(x) + \cdots + u_N(x) t^{r-1}$
    \State $v(x,t) \gets u(x,t) \bar{A} (x,t) \bmod t^r$; \, $B(x,t) \gets v(1/t)t^{r-1}$
    \State \Return $B(x,Q(x,y)) \bmod P(x,y)$
  \end{algorithmic}
\end{algorithm}

We shall first illustrate the connection of \pbseqterm{} and \pbbivmodpow{} by means of an example. Let $P(x,y) = y^2-xy-1$ and $Q(x,y)=y$, i.e., we are looking for $F_{n-1}(x),F_n(x) \in \LL$ such that
\begin{equation} \label{eq:yNfib}
y^n = S(x,y) (y^2-xy-1) + y F_{n}(x) + F_{n-1}(x),
\end{equation}
for some polynomial $S(x,y) \in \LL[y]$. Replace $y$ by $1/y$ in (\ref{eq:yNfib}) and then multiply by $y^{n+1}/(1-xy-y^2)$ to obtain
\[
\frac{y}{1-xy-y^2} = S(x,1/y) y^{n-1} + y^{n} \frac{F_{n}(x) + y F_{n-1}(x)}{1-xy-y^2}.
\]
Now observe that $\deg_y (S(x,1/y) y^{n-1}) \leq n-1$, hence by extracting the $n$th and $(n+1)$st coefficients,
\begin{align*}
F_{n}(x) & = [y^n] \frac{y}{1-xy-y^2}  \quad \text{ and} \\
F_{n-1 }(x) + x F_{n}(x) & = [y^{n+1}] \frac{y}{1-xy-y^2}.
\end{align*}
We conclude that $F_k(x)$ is the $k$th Fibonacci polynomial, for $k=n$ and $k =
n-1$. In particular, each $F_k(x)$ satisfies a linear recurrence with constant
polynomials and can be found in $O(k)$ by \pbseqterm.

This strategy, outlined on an example, generalizes in the obvious way. Explicitly, we have the following lemma (see \cite[Lem.\,2]{BoMo21}).

\begin{lemma} \label{lem:pbbivmodpow1}
    Let $P \in \K[x,y]$ and $r \coloneqq \deg_y(P)$, with $P(x,0) \neq 0$ and
    reversal $\bar{P}(x,y) \coloneqq y^r P(x,\frac{1}{y})$. Write
    \(
        \frac{1}{\bar{P}(x,y)} \eqqcolon \sum_{k \geq 0} u_k(x) y^k.
    \)
    Finally, let $v(x,y) = (u_{N-r+1}(x) + \cdots + u_N(x) y^{r-1}) \bar{P} (x,y) \bmod y^r$. Then $y^N \bmod P(x,y) = v(1/y)y^{r-1}$.
\end{lemma}

The sequence $(u_k(x))_{k\geq0}$ in \cref{lem:pbbivmodpow1} is C-finite because its generating function is rational. Hence, using \pbseqterm, the $r = O(1)$ many terms $u_{N-r+1}(x),\dots,u_{N}(x)$ can be computed in complexity $O(N)$. It follows that the case $Q(x,y)=y$ of \pbbivmodpow{} can be solved in $O(N)$ steps as well.

Finally, the computation of $Q(x,y)^N \bmod P(x,y)$ can be reduced to $y^N \bmod P(x,y)$ with a resultant precomputation (see \cref{lem:bivmodpow_resultant}).
This leads to \cref{alg:bivmodpow}, which solves \pbbivmodpow{} in
\(\bigO{N}\).

\begin{lemma}
  \label{lem:bivmodpow_resultant}
   Let $P(y), Q(y) \in \LL[y]$. Define $A(t), B(t) \in \LL[t]$ by $ A(t) = \mathrm{Res}_y(P(y),t-Q(y))$ and $B(t) = t^N \bmod A(t)$. Then
    \[
        Q(y)^N \bmod P(y) = B(Q(y)) \bmod P(y).
    \]
\end{lemma}
\begin{proof}
    By the definition of the resultant, $A(t) = \prod_i (t-Q(a_i))$ where $a_i \in \overline{\LL}$ are the solutions of $P(y) = 0$. Hence, $P(y)$ divides $A(Q(y))$, which, by construction, divides $B(Q(y)) - Q(y)^N$.
\end{proof}

\subsection{Computing polynomial matrix powers}
\label{sec:modpow_to_matpow}

Let $M \in \pring^{r \times r}$ be an $r \times r$ polynomial matrix of degree
at most $d$. Its power $M^N$ has degree at most $Nd \in O(N)$. Let $P(x,y)$ be
the characteristic polynomial of $M$.  Since $P(x,M) = 0$ by the
Cayley-Hamilton theorem, we get $M^N = R(x,M)$ where $R(x,y) = y^N \bmod
P(x,y)$. The polynomial $R$ can be computed in $O(N)$ via \pbbivmodpow{}. Then
evaluating $R(x,y)$ at $y=M(x)$ has cost $O(N)$ since
$\deg_x(R) \in O(N)$, $\deg_y(R) < r$ with $r \in O(1)$ and $\deg_x(M) = d \in O(1)$.
Hence \cref{alg:polmatpow} is correct and has complexity $O(N)$.

\begin{algorithm}[htb]
	\caption{PolMatPow\((M, N)\)}
	\label{alg:polmatpow}
  \begin{algorithmic}[1]
    \Require{matrix $M(x) \in \K[x]^{r \times r}$, integer $N \in \NN$.}
    \Ensure{$M(x)^N  \in \K[x]^{r \times r}$.}

    \State $P(x,y) \gets$ the characteristic polynomial of $M(x)$
    \State $R(x,y) \gets y^N \bmod P(x,y)$ \Comment{instance of \pbbivmodpow}
    \State \Return $R(x,M(x))$
  \end{algorithmic}
\end{algorithm}


\section{Experiments}
\label{sec:experiments}

The main precomputation step for all our algorithms consists in 
starting with a rational function $U \in \K(x,y) \cap \K\ps{x,y}$
and in finding a
differential operator $L_n$ that annihilates $u_n(x) = [y^n] U(x,y)$ and whose
degree and order are independent of $n$. For this task, in practice, we may
either use the method described in \S\ref{sec:method_algeqtodiffeq}, or
creative telescoping algorithms for hyperexponential functions.
\cref{tab:CT} summarizes timings for a variety of implementations.

The table reveals that, among these implementations, the fastest one for
computing a telescoper of $U(x,y)/y^{n+1}$ is the reduction-based creative
telescoping in Maple. More specifically, \texttt{redct} is the fastest,
followed by HT. The implementation in \emph{ore\_algebra} \cite{KaMe19} in
SageMath competes best with reduction-based methods.

\cref{tab:unroll} gives timings of an efficient implementation of the remaining
stages after precomputations: computing initial terms (IT), and unrolling (UR).
We observe that IT takes negligible time compared to UR, except for extreme
parameter ranges where, simultaneously, \(r\) and \(d\) are large and \(N\) is
small; this is expected since these ranges correspond to cases where the order
of the recurrence to be unrolled is close to \(N\). We also see that binary
powering is always slower, often by a factor more than \(5\), than the addition
of IT and UR. The speed-up factor is summarized in
\cref{fig:speedup_fixedD,fig:speedup_fixedR,fig:speedup_fixedN}; as expected it
grows when \(N\) grows, with \(r\) and \(d\) fixed.

For large \(N\), in most of the reported cases, performing both the
precomputation and IT$+$UR is much faster than using binary powering. Still,
this is not always true, e.g.~for \(r=5\). One has to keep in mind that
\texttt{redct} is not implemented in low-level Maple, and targets rational
coefficients: for a more meaningful assessment of the precomputation part, it
would be interesting to have an implementation of creative telescoping which is
fully optimized and specialized to coefficients in a word-size prime field.


\section{Perspectives}

We have shown that it is possible to beat, both in theory and in practice, the
basic and powerful binary powering method for computing: (i) powers of
polynomial matrices, (ii) terms in polynomial C-finite sequences and (iii)
modular exponentiation for bivariate polynomials. We describe below several
lines of work, including possible optimizations and generalizations, left
for future investigations.

\myparagraph{More detailed complexity analysis}
The most natural direction for future work is to analyze and improve the complexity of
the algorithms in \cref{thm:main} with respect to the parameters~$r$ and~$d$. For simplicity,
these parameters were assumed to be $O(1)$ in this work. For the $N$th power of an \(r \times r\)
matrix $M(x)$ of degree $d$, binary powering has complexity
$O(\timepm{Nd} r^2 + {Nd} r^\omega)$, where $\omega\in[2,3]$ is a feasible exponent of matrix
multiplication over~$\K$. With our approach,
it is legitimate to target a 
differential equation satisfied by the entries of $M(x)^N$
of order $r$ with coefficients in~$x$ of degree~$O(dr^3)$,
yielding a recurrence of order~$O(dr^3)$ and coefficients in~$n$ of degree at most~$r$.
For large~$N$, this would result in a complexity in $O(Ndr^2 \timepm{r})$.
Using different LDEs, of order~$O(r)$ and coefficients  of degree~$O(dr^2)$
could even lead to $O(Ndr \timepm{r})$.

\myparagraph{The $K$th coefficient of the $N$th term} For some (large) integers $N, K \in \NN$, one
might be interested in computing the single coefficient $[x^K y^N] U(x,y)$ of a rational function $U
\in \K(x,y)\cap \K\ps{x,y}$. Equivalently it is natural to wonder: how fast can one compute the $K$th coefficient of the
$N$th term of a C-finite sequence $(u_n(x))_{n \geq 0}$? Using our method, a recurrence with
initial conditions for the coefficients of $u_N(x)$ can be deduced in $O(\log(N))$ 
operations. Then (assuming that the recurrence is nonsingular) the $K$th coefficient can be found
in $O(\timepm{\sqrt{K}})$ operations by using baby-steps/giant-steps techniques~\cite{ChuChu88,BoGaSc07}. 
We expect that, at least under a genericity assumption, this problem can be solved in 
complexity $O(\log(N) + \timepm{\sqrt{K}})$ which is a big improvement compared to the previous
best $O(N+K)$ by Massazza and Radicioni~\cite{MaRa05}. 

\myparagraph{Polynomial P-finite sequences} A somewhat related task is to study the analogous
problem to \pbseqterm{} for polynomial P-finite sequences, that is for $(u_n(x))_n \in \K[x]^\mathbb{N}$
satisfying
    \[
    p_r(x,n)u_{n+r}(x) + \cdots + p_0(x,n)u_n(x) = 0,
    \]
for $p_i(x,n) \in \K[x,n]$. We expect that, at least under a genericity assumption, a
generalization of \cref{lem:hermiteCT} (based on results in~\cite{BoChLaSa18,Hoeven21}) should exist,
implying in particular that $u_N(x)$ satisfies an LDE of order and degree independent of~$N$.
Generalizing this even further, one might study the creative telescoping problem for rational
functions of the form
    \(
        H(\textbf{x}) = \frac{P(x_1,\dots,x_s)}{Q(x_1,\dots,x_s) R(x_1,\dots,x_s)^n}.
    \)
We expect that (at least generically) the minimal telescoper for $H(\textbf{x})$ has order and
degree independent of $n$ and can be found via a Griffiths-Dwork reduction type approach, based on
ideas from~\cite{BoLaSa13}.

\myparagraph{Connection to the Jordan–Chevalley decomposition} A different approach for computing
powers of matrices uses the \emph{Jordan–Chevalley decomposition} (also called \emph{SN
decomposition}), see e.g.~\cite{HsKoSi96,ElHa98,Schmidt00,CoEsZa11}. It ensures that any polynomial
matrix $M \in \K[x]^{r \times r}$ can be written as
    \(
        M = S + Z,
    \)
where $S\in \K(x)^{r \times r}$ is diagonalizable over $\overline{\K(x)}$, $Z\in \K(x)^{r
\times r}$ is nilpotent, and $SZ=ZS$. From this decomposition it follows that $M^N =
\sum_{i=0}^{r-1} \binom{N}{i} S^{N-i} Z^{i}$. After a change of basis, this reduces to computing a
power of a diagonal matrix with algebraic functions coefficients. Using \cref{lem:poweralgeqdiffeq}
this can be performed efficiently in $O(N)$ operations. It would be certainly interesting to
compare this approach with the other methods.

\myparagraph{A PDE approach for \pbseqterm} 
There is yet another method to deduce recurrence~\eqref{eq-recFn}. 
The starting point is that the generating function $F(x,y) = y/(1-xy-y^2)$ of $F_n(x)$ satisfies the linear PDE
\begin{equation}\label{eq:pdeF}
	\textstyle
  (x^2+4) \frac{\partial^2 F}{\partial x^2} +
3x\frac{\partial F}{\partial x}  - y^2 \frac{\partial^2 F}{\partial y^2} - y\frac{\partial F}{\partial y} + F = 0,
\end{equation}
and extracting the coefficient of $x^k y^n$ in~\eqref{eq:pdeF} immediately gives~\eqref{eq-recFn}.
More generally, such a PDE translates into a recurrence 
if it is linear with polynomial coefficients in $x$ and $y$ 
and if additionally all terms of the form $x^i y^\ell
\frac{\partial^k F}{\partial x^k} \frac{\partial^j F}{\partial y^j}$ have $\ell = j$. A dimension
counting argument in the spirit of~\cite[Lem.\,3]{Lipshitz88} proves that such a PDE exists for
\emph{any} rational function $F(x,y)$. The existence proof is effective and amounts to linear
algebra. 
A natural
question is whether it is possible to compute such a PDE via creative telescoping
(either Almkvist-Zeilberger~\cite{AlZe90} or reduction-based~\cite{BoChChLi10}),
and how the corresponding method compares to the aforementioned ones.

\myparagraph{Integer case in bit complexity $O(N)$}
Recall the analogy between the bit complexity for finding the $N$th term of a
C-finite sequence over $\ZZ$ and the complexity for finding the $N$th term of a
C-finite sequence over \(\pring\). Our work achieves $O(N)$ for the latter, so
it is now natural to target $O(N)$ for the former, for instance for the $N$th
Fibonacci number or simply \(3^N\). This remains widely open.

\onecolumn
\begin{table}[h]
  \caption{%
    \textmd{%
      Timings in seconds for creative telescoping to find a
      telescoper $L_n$ of $P(x,y)/(y^{n+1} Q(x,y))$. Here \(P(x,y)/Q(x,y)\) is
      the generating function for the sequence of the top-right entry of the
      powers of a randomly chosen matrix in $\mathbb{F}_p[x]^{r \times r}$ of
      degree $d$, for a 50-bit prime \(p\), with $Q(x,y)$ the \(y\)-reversal
      of the characteristic polynomial of this matrix. The order of $L_n$ is
      $\ell$, its degree in $n$ is $\mathrm{d}_n$, and $\mathrm{d}_x =
      \deg_x(L_n)$. A blank space means that the computation took more than 1000 seconds.
      We observe empirically that the degree in $x$ is
      $dr(r+1)(2r-1)/2 - r(r-1)$ while its degree in $n$ is $(r-1)(r+2)/2$;
      this is expected asymptotically by~\cite[Thm.\,25]{BoChChLi10} and
      \cref{lem:hermiteCT}, because $\deg_y Q(x,y) = r$ and $\deg_x Q(x,y) =
      dr$.
      The tested implementations are \\
      \textbullet~~ in Maple: \texttt{redct} \cite{BoChLaSa18}; 
                              \texttt{HermiteTelescoping} (HT) \cite{BoLaSa13};
                              \texttt{Zeilberger} (ZB) \cite{AlZe90} in DEtools;
                              \texttt{creative\_telescoping} (c\_t) \cite{Chyzak00}; \\
      \textbullet~~ in SageMath: creative telescoping (\texttt{ct}) from the \emph{ore\_algebra} package \cite{KaMe19}; \\
      \textbullet~~ In Mathematica: \texttt{FindCreativeTelescoping} (FCT),
                                    \texttt{CreativeTelescoping} (CT), and
                                    \texttt{HermiteTelescoping} (HCT), see \cite{Koutschan10}.
    }
  }
  \label{tab:CT}
  \begin{tabular}{ll||c|c|c|c|c|c|c|c||c|c|c}
    &     & \multicolumn{4}{c|}{Maple} & Sage & \multicolumn{3}{c||}{Mathematica} & $\ell$ & $\mathrm{d}_n$ & $\mathrm{d}_x$ \\
    $r$ & $d$ & redct & HT  & ZB  & c\_t & ct  & FCT & CT   & HCT  &   &    &       \\ \hline \hline
        & 2   & 0.0   & 0.1 & 0.0 & 0.1  & 0.5 & 0.2 & 0.2  & 0.2 & 2 & 2  & 16    \\ \cline{2-13}
    2   & 4   & 0.0   & 0.0 & 0.0 & 0.1  & 0.6 & 0.4 & 0.4  & 0.3 & 2 & 2  & 34    \\ \cline{2-13}
        & 6   & 0.0   & 0.0 & 0.0 & 0.1  & 0.6 & 0.7 & 0.5  & 0.5 & 2 & 2  & 52    \\ \cline{2-13}
        & 8   & 0.0   & 0.0 & 0.0 & 0.1  & 0.8 & 1.0 & 0.7  & 0.7 & 2 & 2  & 70    \\ \hline \hline
        & 1   & 0.0   & 0.2 & 0.0 & 0.5  & 2.0 & 2.0 & 1.3  & 1.3 & 3 & 5  & 24    \\ \cline{2-13}
        & 2   & 0.0   & 0.1 & 0.8 & 3.4  & 3.1 & 4.0 & 2.6  & 2.5 & 3 & 5  & 54    \\ \cline{2-13}
    3   & 3   & 0.1   & 0.2 & 0.8 & 9.3  & 5.6 & 10  & 5.7  & 5.4 & 3 & 5  & 84    \\ \cline{2-13}
        & 4   & 0.1   & 0.5 & 18  & 19   & 8.2 & 17  & 9.4  & 8.9 & 3 & 5  & 114   \\ \cline{2-13}
        & 5   & 0.2   & 1.1 & 5.1 & 32   & 12  & 25  & 14   & 14  & 3 & 5  & 144   \\ \cline{2-13}
        & 6   & 0.5   & 1.7 & 9.8 & 49   & 17  & 35  & 19   & 20  & 3 & 5  & 174   \\ \hline \hline
        & 1   & 0.4   & 2.9 & 23  & 117  & 20  & 31  & 25   & 25  & 4 & 9  & 58    \\ \cline{2-13}
        & 2   & 1.7   & 17  & 410 & 749  & 45  & 101 & 96   & 95  & 4 & 9  & 128   \\ \cline{2-13}
    4   & 3   & 4.4   & 43  &     &      & 89  & 295 & 376  & 373 & 4 & 9  & 198   \\ \cline{2-13}
        & 4   & 12    & 82  &     &      & 172 & 388 & 752  & 693 & 4 & 9  & 268   \\ \cline{2-13}
        & 5   & 18    & 128 &     &      & 280 & 635 &      &     & 4 & 9  & 338   \\ \hline \hline
        & 1   & 11    & 34  & 538 &      & 163 & 847 & 780  &     & 5 & 14 & 115   \\ \cline{2-13}
    5   & 2   & 64    & 183 &     &      & 515 &     &      &     & 5 & 14 & 250   \\ \cline{2-13}
        & 3   & 159   & 526 &     &      &     &     &      &     & 5 & 14 & 385   \\ \cline{2-13}
        & 4   & 345   &     &     &      &     &     &      &     & 5 & 14 & 520   \\
  \end{tabular}
\end{table}

\begin{table}[h]
  \caption{%
    \textmd{%
      Timings in seconds, using the C++ library NTL \cite{ShoupNTL} and PML
      \cite{HyunNeigerSchost2019}, for computing the top-right entry of the
      \(N\)th power of a randomly chosen matrix in $\mathbb{F}_p[x]^{r \times
      r}$ of degree $d$, for a 50-bit prime \(p\). The first tested method is
      to directly apply binary powering (BP); in the present context, the
      polynomial matrix multiplication of PML is based on
      evaluation-interpolation and 3-prime FFT. The second tested method uses
      \cref{alg:seqtermct} and we do not count ``precomputations'' (already
      showed in \cref{tab:CT}), i.e.~we only report timings for the two
      non-negligible steps that depend on \(N\), namely \cref{step:seqtermas:unroll}
      (UR, unrolling) and \cref{step:seqtermas:init_values} (IT, initial terms)
      from \cref{alg:seqtermas}.
    }
  }
  \label{tab:unroll}
  \scalebox{0.7}{
  \begin{tabular}{ll||c|c|c||c|c|c||c|c|c||c|c|c||c|c|c||c|c|c||c|c|c}
    &     & \multicolumn{3}{c||}{\(N=2^{10}\)} & \multicolumn{3}{c||}{\(N=2^{12}\)} & \multicolumn{3}{c||}{\(N=2^{14}\)}
          & \multicolumn{3}{c||}{\(N=2^{16}\)} & \multicolumn{3}{c||}{\(N=2^{18}\)} & \multicolumn{3}{c||}{\(N=2^{20}\)}
          & \multicolumn{3}{c}{\(N=2^{22}\)} 
          \\
    $r$ & $d$ &   BP   &   UR   &   IT   &   BP   &   UR   &   IT   &   BP   &   UR   &   IT   &   BP   &   UR   &   IT   &   BP   &   UR   &   IT   &   BP   &   UR   &   IT   &   BP   &  UR    &   IT   \\ \hline \hline
        & 2   & 1.2e-3 & 5.7e-4 & 3.7e-5 & 5.3e-3 & 2.4e-3 & 4.3e-5 & 2.5e-2 & 9.7e-3 & 4.9e-5 & 1.1e-1 & 3.9e-2 & 5.5e-5 & 5.3e-1 & 1.5e-1 & 6.2e-5 & 3.3e+0 & 6.2e-1 & 6.7e-5 & 1.5e+1 & 2.5e+0 & 7.5e-5 \\ \cline{2-23}
    2   & 4   & 2.6e-3 & 1.3e-3 & 7.8e-5 & 1.2e-2 & 5.2e-3 & 9.4e-5 & 5.2e-2 & 2.1e-2 & 1.1e-4 & 2.4e-1 & 8.4e-2 & 1.3e-4 & 1.4e+0 & 3.4e-1 & 1.4e-4 & 7.2e+0 & 1.4e+0 & 1.6e-4 & 3.1e+1 & 5.4e+0 & 1.8e-4 \\ \cline{2-23}
        & 6   & 3.8e-3 & 2.1e-3 & 1.2e-4 & 1.7e-2 & 8.7e-3 & 1.5e-4 & 7.9e-2 & 3.5e-2 & 1.8e-4 & 3.6e-1 & 1.4e-1 & 2.1e-4 & 2.3e+0 & 5.5e-1 & 2.4e-4 & 1.0e+1 & 2.2e+0 & 2.7e-4 & 4.6e+1 & 8.9e+0 & 3.0e-4 \\ \cline{2-23}
        & 8   & 5.3e-3 & 3.1e-3 & 1.9e-4 & 2.4e-2 & 1.2e-2 & 2.4e-4 & 1.1e-1 & 5.0e-2 & 2.8e-4 & 5.3e-1 & 2.0e-1 & 3.3e-4 & 3.3e+0 & 8.0e-1 & 3.8e-4 & 1.5e+1 & 3.2e+0 & 4.3e-4 & 7.0e+1 & 1.2e+1 & 4.9e-4 \\ \hline \hline
        & 1   & 1.4e-3 & 3.0e-4 & 1.3e-4 & 6.0e-3 & 1.3e-3 & 1.7e-4 & 2.6e-2 & 5.5e-3 & 2.1e-4 & 1.2e-1 & 2.2e-2 & 2.4e-4 & 5.8e-1 & 8.8e-2 & 2.8e-4 & 3.4e+0 & 3.5e-1 & 3.1e-4 & 1.6e+1 & 1.4e+0 & 3.5e-4 \\ \cline{2-23}
        & 2   & 2.9e-3 & 7.8e-4 & 4.0e-4 & 1.2e-2 & 3.2e-3 & 5.3e-4 & 5.6e-2 & 1.3e-2 & 6.5e-4 & 2.6e-1 & 5.2e-2 & 7.8e-4 & 1.5e+0 & 2.1e-1 & 9.1e-4 & 7.6e+0 & 8.4e-1 & 1.0e-3 & 3.4e+1 & 3.3e+0 & 1.2e-3 \\ \cline{2-23}
    3   & 3   & 4.3e-3 & 1.4e-3 & 7.4e-4 & 1.9e-2 & 5.8e-3 & 9.9e-4 & 8.4e-2 & 2.3e-2 & 1.2e-3 & 3.9e-1 & 9.3e-2 & 1.5e-3 & 2.2e+0 & 3.7e-1 & 1.7e-3 & 1.1e+1 & 1.5e+0 & 2.0e-3 & 4.9e+1 & 6.0e+0 & 2.2e-3 \\ \cline{2-23}
        & 4   & 6.0e-3 & 2.1e-3 & 8.0e-4 & 2.6e-2 & 8.8e-3 & 1.0e-3 & 1.2e-1 & 3.5e-2 & 1.3e-3 & 5.8e-1 & 1.4e-1 & 1.5e-3 & 3.5e+0 & 5.7e-1 & 1.8e-3 & 1.7e+1 & 2.3e+0 & 2.0e-3 & 7.1e+1 & 9.1e+0 & 2.3e-3 \\ \cline{2-23}
        & 5   & 7.4e-3 & 3.0e-3 & 1.0e-3 & 3.3e-2 & 1.2e-2 & 1.3e-3 & 1.5e-1 & 5.0e-2 & 1.7e-3 & 7.2e-1 & 2.0e-1 & 2.0e-3 & 4.3e+0 & 7.9e-1 & 2.3e-3 & 2.0e+1 & 3.2e+0 & 2.6e-3 & 8.8e+1 & 1.3e+1 & 2.9e-3 \\ \cline{2-23}
        & 6   & 9.1e-3 & 4.0e-3 & 1.2e-3 & 4.0e-2 & 1.6e-2 & 1.6e-3 & 1.8e-1 & 6.6e-2 & 1.9e-3 & 8.2e-1 & 2.7e-1 & 2.3e-3 & 5.3e+0 & 1.1e+0 & 2.7e-3 & 2.3e+1 & 4.2e+0 & 3.1e-3 & 1.1e+2 & 1.7e+1 & 3.4e-3 \\ \hline \hline
        & 1   & 2.7e-3 & 4.2e-4 & 7.8e-4 & 1.1e-2 & 1.8e-3 & 1.1e-3 & 4.9e-2 & 7.5e-3 & 1.4e-3 & 2.2e-1 & 3.0e-2 & 1.7e-3 & 1.1e+0 & 1.2e-1 & 2.0e-3 & 6.2e+0 & 4.8e-1 & 2.3e-3 & 2.9e+1 & 1.9e+0 & 2.6e-3 \\ \cline{2-23}
        & 2   & 5.5e-3 & 1.2e-3 & 1.3e-3 & 2.4e-2 & 5.2e-3 & 1.8e-3 & 1.1e-1 & 2.1e-2 & 2.3e-3 & 4.9e-1 & 8.6e-2 & 2.8e-3 & 2.8e+0 & 3.4e-1 & 3.2e-3 & 1.4e+1 & 1.4e+0 & 3.7e-3 & 6.2e+1 & 5.5e+0 & 4.2e-3 \\ \cline{2-23}
    4   & 3   & 8.2e-3 & 2.4e-3 & 2.1e-3 & 3.6e-2 & 1.0e-2 & 2.9e-3 & 1.6e-1 & 4.2e-2 & 3.7e-3 & 7.3e-1 & 1.7e-1 & 4.5e-3 & 4.4e+0 & 6.7e-1 & 5.3e-3 & 2.1e+1 & 2.7e+0 & 6.1e-3 & 9.3e+1 & 1.1e+1 & 6.9e-3 \\ \cline{2-23}
        & 4   & 1.1e-2 & 4.1e-3 & 3.0e-3 & 5.0e-2 & 1.7e-2 & 4.1e-3 & 2.3e-1 & 6.9e-2 & 5.2e-3 & 1.1e+0 & 2.8e-1 & 6.4e-3 & 6.6e+0 & 1.1e+0 & 7.5e-3 & 3.1e+1 & 4.5e+0 & 8.6e-3 & 1.3e+2 & 1.8e+1 & 9.7e-3 \\ \cline{2-23}
        & 5   & 1.4e-2 & 6.0e-3 & 3.7e-3 & 6.3e-2 & 2.5e-2 & 5.2e-3 & 2.8e-1 & 1.0e-1 & 6.6e-3 & 1.3e+0 & 4.1e-1 & 8.0e-3 & 7.7e+0 & 1.6e+0 & 9.4e-3 & 3.6e+1 & 6.5e+0 & 1.1e-2 & 1.6e+2 & 2.6e+1 & 1.2e-2 \\ \hline \hline
        & 1   & 4.4e-3 & 6.0e-4 & 1.8e-3 & 1.8e-2 & 2.7e-3 & 2.5e-3 & 8.2e-2 & 1.1e-2 & 3.3e-3 & 3.7e-1 & 4.5e-2 & 4.1e-3 & 1.7e+0 & 1.8e-1 & 4.9e-3 & 1.0e+1 & 7.3e-1 & 5.7e-3 & 4.7e+1 & 2.9e+0 & 6.5e-3 \\ \cline{2-23}
    5   & 2   & 9.1e-3 & 2.0e-3 & 3.5e-3 & 3.9e-2 & 9.1e-3 & 5.1e-3 & 1.8e-1 & 3.7e-2 & 6.7e-3 & 8.1e-1 & 1.5e-1 & 8.3e-3 & 4.6e+0 & 6.0e-1 & 9.9e-3 & 2.3e+1 & 2.4e+0 & 1.2e-2 & 1.0e+2 & 9.6e+0 & 1.3e-2 \\ \cline{2-23}
        & 3   & 1.3e-2 & 4.3e-3 & 5.7e-3 & 5.8e-2 & 1.9e-2 & 8.3e-3 & 2.6e-1 & 7.8e-2 & 1.1e-2 & 1.2e+0 & 3.2e-1 & 1.4e-2 & 7.1e+0 & 1.3e+0 & 1.6e-2 & 3.4e+1 & 5.1e+0 & 1.9e-2 & 1.5e+2 & 2.0e+1 & 2.2e-2 \\ \cline{2-23}
        & 4   & 1.8e-2 & 7.4e-3 & 7.8e-3 & 8.0e-2 & 3.3e-2 & 1.2e-2 & 3.8e-1 & 1.3e-1 & 1.5e-2 & 1.8e+0 & 5.4e-1 & 1.9e-2 & 1.1e+1 & 2.2e+0 & 2.3e-2 & 4.9e+1 & 8.7e+0 & 2.6e-2 & 2.1e+2 & 3.5e+1 & 3.0e-2 \\
  \end{tabular}
  }
\end{table}

\twocolumn
\begin{figure}[ht]
  \setlength\abovecaptionskip{0pt}
  \setlength\belowcaptionskip{-8pt}
  \centering
  \includegraphics[height=0.207\textheight]{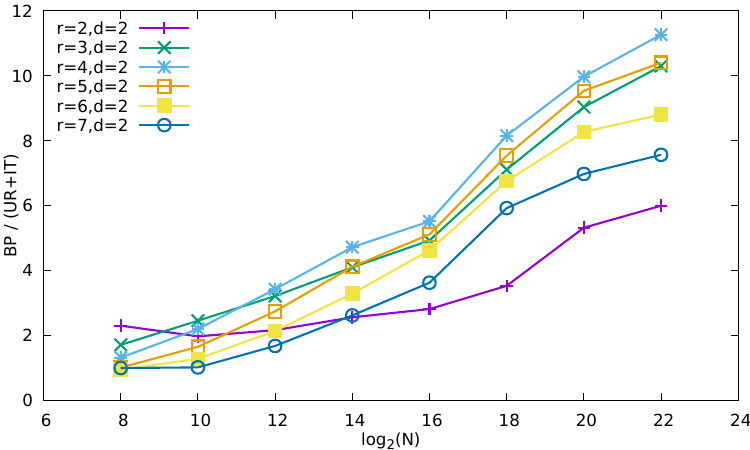}
  \caption{\textmd{Speed-up versus binary powering, not counting
      precomputations, for \(r = 2 \ldots 7\), \(N=2^8,2^{10},\ldots,2^{22}\), and fixed \(d=2\).}}
  \label{fig:speedup_fixedD}
\end{figure}
\begin{figure}[ht]
  \setlength\abovecaptionskip{0\baselineskip}
  \setlength\belowcaptionskip{-8pt}
  \centering
  \includegraphics[height=0.207\textheight]{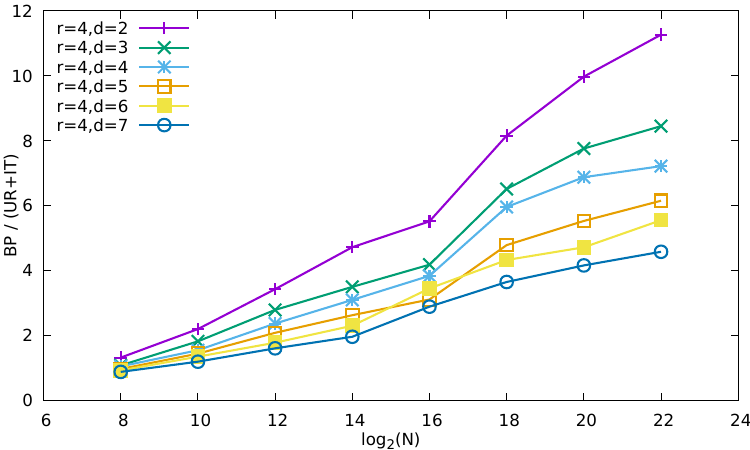}
  \caption{\textmd{Speed-up versus binary powering, not counting
      precomputations, for \(d = 2 \ldots 7\), \(N=2^8,2^{10},\ldots,2^{22}\), and fixed \(r=4\).}}
  \label{fig:speedup_fixedR}
\end{figure}
\begin{figure}[ht]
  \setlength\abovecaptionskip{0\baselineskip}
  \setlength\belowcaptionskip{-8pt}
  \centering
  \includegraphics[height=0.207\textheight]{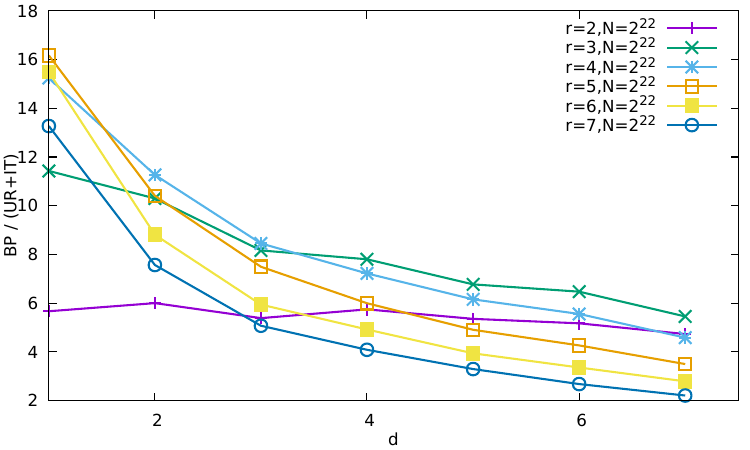}
  \caption{\textmd{Speed-up versus binary powering, not counting
      precomputations, for \(r = 2 \ldots 7\), \(d = 1 \ldots 7\), and fixed \(N=2^{22}\).}}
  \label{fig:speedup_fixedN}
\end{figure}


\balance


\end{document}